\documentclass[11pt]{article}
\usepackage{amsmath,amssymb,amsthm}
\usepackage[margin=1in]{geometry}
\usepackage{url}
\usepackage{hyperref}
\usepackage{amsrefs}

\renewcommand{\>}{\rangle}
\newcommand{\<}{\langle}
\newcommand{\norm}[1]{\|#1\|}
\DeclareMathOperator{\poly}{\operatorname{poly}}
\DeclareMathOperator{\tr}{\operatorname{tr}}
\newcommand{\F}[1]{\mathbb{F}_{#1}}
\newcommand{\PF}[1]{\mathbb{PF}_{#1}}
\newcommand{\level}[2]{L_{#1,#2}}
\newcommand{\sphere}[1]{\mathcal{S}_{#1}}
\newcommand{\dist}[1]{\Delta(#1)}
\newcommand{\quadchar}{\chi}
\newcommand{\kloossym}[1]{K_{#1}}
\newcommand{\kloos}[3]{\kloossym{#1}(#2,#3)}

\newcommand{\be}{\begin{equation}}
\newcommand{\ee}{\end{equation}}
\def\ba#1\ea{\begin{align}#1\end{align}}
\newcommand{\nn}{\nonumber\\}
\newcommand{\Span}{{\mathrm{affspan}}}

\newtheorem{theorem}{Theorem}
\newtheorem{corollary}[theorem]{Corollary}
\newtheorem{lemma}[theorem]{Lemma}

\newcommand{\eq}[1]{(\ref{eq:#1})}
\renewcommand{\sec}[1]{Section~\ref{sec:#1}}
\newcommand{\app}[1]{Appendix~\ref{app:#1}}
\newcommand{\thm}[1]{Theorem~\ref{thm:#1}}
\newcommand{\cor}[1]{Corollary~\ref{cor:#1}}
\newcommand{\lem}[1]{Lemma~\ref{lem:#1}}

\begin{document}

\title{Quantum algorithms for hidden nonlinear structures}
\author{Andrew M.\ Childs \\ \small amchilds@caltech.edu
\and Leonard J.\ Schulman \\ \small schulman@caltech.edu
\and Umesh V.\ Vazirani \\ \small vazirani@cs.berkeley.edu}
\date{}

\maketitle


\begin{abstract}
Attempts to find new quantum algorithms that outperform classical
computation have focused primarily on the nonabelian hidden subgroup
problem, which generalizes the central problem solved by Shor's
factoring algorithm.  We suggest an alternative generalization, namely
to problems of finding hidden nonlinear structures over finite fields.
We give examples of two such problems that can be solved efficiently
by a quantum computer, but not by a classical computer.  We also give
some positive results on the quantum query complexity of finding
hidden nonlinear structures. 
\end{abstract}

\section{Introduction}
\label{sec:intro}

One of the major open problems in quantum computation is to develop new quantum algorithms. Much of the work on this question has focused on the nonabelian hidden subgroup problem (HSP), attempting to extend the quantum solution of the abelian HSP \cites{Sim94,Sho97,Kit97}. Unfortunately, these efforts have met with only limited success. In this paper, we describe an alternative way of generalizing the success of Shor's algorithm.

The key to exponential savings in quantum algorithms is the creation of sharp constructive interference in large sets.  Such precise interference is only known to arise in a few cases, primarily in which the set is a group.  Under these conditions, the key to quantum speed-up is to diagonalize the group algebra, i.e., to perform a Fourier transform.  Once this has been done, certain structures become easy to detect.

The structures that have been investigated so far are subgroups and their cosets. In the case of abelian groups, the Fourier transform is a mapping from the group to its dual, and this mapping respects subgroups and cosets. Advances in quantum algorithms have been pursued by extending the groups from abelian to nonabelian, but in the nonabelian case there is no dual group, and the same approach is not available. Indeed, certain methods that work in the abelian case are known to fail in some nonabelian cases, such as the symmetric group~\cites{GSVV01,MRS05,HMRRS06}.

Our approach in this paper is to shift the focus back to the Fourier transform over abelian groups, and to consider what other hidden structures can be revealed by abelian Fourier transforms via constructive interference effects. We turn for inspiration to optics and acoustics, where light or sound can be highly focused (i.e., undergo highly constructive interference) when reflected by a conic (e.g., parabolic or elliptic) surface.
To connect this idea to known quantum algorithms, observe that abelian hidden subgroup problems, when restricted to a vector space, can be viewed as determining a hidden linear structure.
(Most generally, this viewpoint makes sense for a module over any ring, but here we restrict ourselves to vector spaces over finite fields.)
Any subgroup of the additive group of $\F{q}^d$ ($q=p^m$ a prime power) is an $\F{q}$-linear subspace, and the cosets of this subgroup consist of parallel affine subspaces, or \emph{flats}.  Given a black box function that is constant on each flat and distinct on different ones, abelian Fourier sampling determines the hidden subspace in time $\poly(d \log q)$.
Pursuing the analogy with wave mechanics, our approach is to set up black box functions that are constant on quadratic surfaces, and use interference effects to discover properties of the unknown quadratic.  More generally, we will study this approach for algebraic sets of higher degree.

The first problem we study is the \textit{hidden radius problem}.  In this problem, the hidden property is the radius $r$ of a sphere.  We give an efficient quantum algorithm for determining one bit of $r$, namely whether or not it is a quadratic residue, assuming that the dimension is odd.  With a classical computation, even this restricted problem requires exponentially many queries.  (For the problem of determining the other bits of $r$, we argue that the quantum query complexity is small.)

The second problem we discuss is the \textit{hidden flat of centers problem}. In this problem, the radius of the sphere is fixed (say, at $r=1$), but its center is constrained to lie in an unknown flat in $\F{q}^d$.  For example, the centers of the spheres may lie on an unknown line.  For this problem, we give an efficient quantum algorithm to determine the entire hidden flat, not just one bit of information about it. However, this algorithm also works only when the dimension is odd. The main idea of the algorithm is to use a quantum walk to move amplitude from the spheres to their centers. Our algorithms for both this and the hidden radius problem make crucial use of a connection to certain exponential sums called \emph{twisted Kloosterman sums}.

Both of the above problems fall into a framework of \textit{shifted subset problems}. For problems in this class, the main idea is to define a black box function that is constant on some subset of the points in $\F{q}^d$, as well as on shifted versions of this subset, with the function taking distinct values when the shifts are different. The goal may be either to determine some property of the basic subset, or of the allowed shifts, or both. Typically, this will not give a well-defined black box, since different shifts of the subset may lead to overlapping points.  However, we can resolve this issue by defining the black box carefully.

We also obtain results regarding hidden polynomial structures of higher degree.  These results are purely information-theoretic (i.e., regard query complexity).
We introduce the framework of \textit{hidden polynomial problems}. In these problems, the hidden object is a multivariate polynomial $h(x) \in \F{q}[x_1,\ldots,x_d]$ chosen from some set of possible polynomials. We are given a black box function that is constant on the level sets of $h(x)$ (i.e., the sets $\{x \in \F{q}^d: h(x)=y\}$ for various $y \in \F{q}$) and distinct on different level sets, and the goal is to determine $h(x)$.
When $h(x)$ is linear, this is the abelian HSP described above, whereas for more general polynomials, it is typically not an HSP in any group.  (Observe that a hidden polynomial problem, unlike a shifted subset problem, is automatically an oracle problem.)

Assuming the dimension $d$ and the degree of $h(x)$ is constant, we show that the query complexity of the hidden polynomial problem is typically $\poly(\log q)$. We show this by considering an analog of the standard approach to the HSP, wherein one query of the black box is used to produce a quantum state that depends on the hidden object. Provided these states are sufficiently statistically distinguishable, it follows that $\poly(\log q)$ copies contain enough information to determine the hidden object with high probability.
To establish distinguishability of the states, we give two simple but apparently new results about the fidelity between general quantum states satisfying certain intersection conditions, and we show that one of these conditions is satisfied by typical polynomials.  These lemmas could also have applications to problems involving quantum states derived from combinatorial designs that are unrelated to polynomials.

\section{Hidden radius problem}
\label{sec:hrp}

We begin by considering the first of two shifted subset problems, the \emph{hidden radius problem}.
In the quantum version of the hidden radius problem, our goal is to determine an unknown radius $r \in \F{q}$ given a uniform superposition over points in $\F{q}^d$ on a sphere of radius $r$ whose center is chosen uniformly at random.  We give an efficient quantum algorithm for determining whether $r$ is a quadratic residue, provided $d$ is odd.  We also show that the quantum query complexity of finding $r$ is $\poly(\log q)$, again assuming $d$ is odd, and we give evidence that this should also be the case for $d$ even.

For this problem to make sense classically as well as quantumly, we define it in terms of a black box function.
Roughly speaking, we would like to define a black box function that on input $x$, a point on the sphere of radius $r$ with center $t$, outputs some encryption of $t$, thereby giving a function that is constant on shifted spheres and distinct on different spheres.  But this cannot be done directly, since spheres can intersect.
To resolve this issue, we note that the vector $s=x-t$ pointing to $x$ from the center $t$ uniquely describes a particular sphere.  So our black box $f_1$ takes as input the pair $x$ and an encryption $\sigma$ of $s$ and outputs an encryption of $t$.  We also supply a black box $f_{-1}$ that takes as input the pair $x$ and encryption of $t$ and outputs the encryption of $s$.  The goal of the problem is to determine $r$ using an oracle that computes either $f_1$ or $f_{-1}$ as desired.
(In \app{ssp}, we give a black-box formulation of general shifted subset problems, which provides an alternative oracle for the hidden radius problem.)

It is straightforward to show that this problem is hard for a classical computer.

\begin{theorem}\label{thm:classically_hard}
Any classical computation with access to $f_1$ and $f_{-1}$ requires an expected exponential number of queries to obtain a $1/\poly(d \log q)$ bias for any single bit of information about $r$.
\end{theorem}
\begin{proof}[Proof sketch]
Let the hidden radius $r$ be uniformly random, and let $f_1(x,\sigma)$ be a uniformly random one-to-one function of the sphere center $t=x-s$, where $\sigma$ is the encryption of $s$.  Now we can assume without loss of generality that the algorithm is deterministic. Given any sequence of evaluations of $f_1$ and $f_{-1}$ that do not involve any sphere twice, the conditional distribution on subsequent evaluations involving points on other spheres is uniform.  The probability of success is therefore sub-polynomial so long as the square of the number of queries is less than a polynomial fraction of $q^d$.
\end{proof}

To solve this problem on a quantum computer, we use the following state generation procedure.
Begin with a uniform superposition over $x$ and $\sigma$, then compute $f_1$, then uncompute $\sigma$ using $f_{-1}$, and finally discard the function value, giving (up to normalization)
\ba
\sum_{x,\sigma} |x, \sigma\>
&\mapsto \sum_{x,\sigma} |x,\sigma,f_1(x,\sigma)\> \\
&\mapsto \sum_{x,\sigma} |x,f_1(x,\sigma)\> \\
&\mapsto |\sphere{r}+t\> \text{~where $t$ is uniformly random in $\F{q}^d$}
\ea
where $\sphere{r}$ denotes the sphere of radius $r$ centered at the origin, and where we use the convention that for a finite set $S$, $|S\> := \sum_{s \in S} |s\>/\sqrt{|S|}$ denotes the normalized uniform superposition over elements of $S$.
In other words, we can use two queries to the oracle to produce the mixed quantum state
\be
\rho_r := \frac{1}{q^d} \sum_{t \in \F{q}} |\sphere{r}+t\>\<\sphere{r}+t|
\ee
from which we would like to extract information about the hidden radius $r$.

The sphere of radius $r$ centered at the origin is defined by $\sphere{r} := \level{\dist{x}}{r}$, where $\dist{x}:=\sum_{j=1}^d x_j^2$, and where  $\level{f}{y} := f^{-1}(y) = \{x \in X: f(x)=y\}$ denotes the level set of $f(x)$ with value $y$.  The quadratic polynomial $\dist{x}$ can be thought of as measuring the distance from the origin in $\F{q}^d$ (although of course it does not satisfy a triangle inequality); then $\sphere{r}$ consists of the points at distance $r$ from the origin.

Note that since we are working in a finite field, there is no concept of large spheres or small spheres; indeed all spheres contain approximately the same number of points.  In particular, the number of points on the sphere of radius $r$ is \cite{MMST96}*{Theorem 1}
\be
  |\sphere{r}| =
  \begin{cases}
    q^{d-1} + \quadchar((-1)^{(d-1)/2} r) \sqrt{q^{d-1}}
      & d \text{~odd}, r \ne 0 \\
    q^{d-1} - \quadchar((-1)^{d/2}) \sqrt{q^{d-2}}
      & d \text{~even}, r \ne 0 \\
    q^{d-1}
      & d \text{~odd}, r=0 \\
    q^{d-1} + \quadchar((-1)^{d/2})(q-1)\sqrt{q^{d-2}}
      & d \text{~even}, r=0 \,,
  \end{cases}
\label{eq:spheresizes}
\ee
where $\quadchar$ denotes the quadratic character of $\F{q}^\times$.
In other words, up to small corrections, every sphere has about $q^{d-1}$ points on it (except that the sphere of zero radius in two dimensions consists of $2q-1$ points when $q = 1 \bmod 4$; and is simply a single point, the origin, when $q=3 \bmod 4$).

Our goal is to determine $r$ using polynomially many copies of the hidden radius state $\rho_r$.
For any $r$, the state is invariant under arbitrary translations in $\F{q}^d$.  This symmetry can be exploited using the $d$-dimensional Fourier transform over $\F{q}$,
\be
  U := \frac{1}{\sqrt{q^d}} \sum_{x,k \in \F{q}^d} 
  \omega_p^{\tr k \cdot x} |k\>\<x|
\,,
\label{eq:fourier}
\ee
where $\omega_p := e^{2 \pi i/p}$, $k \cdot x := \sum_{j=1}^d k_j x_j$, and where $\tr a := a + a^p + \cdots + a^{q/p}$ denotes the trace from $\F{q}$ to $\F{p}$.  Fourier transforming the state, we find
\be
  U \rho_r U^\dag = \sum_{k \in \F{q}^d} \Pr(k|r) \, |k\>\<k|
  \qquad\text{with}\qquad
  \Pr(k|r)
  = \frac{1}{q^d |\sphere{r}|}
    \left|\sum_{x \in \sphere{r}} \omega_p^{\tr{k \cdot x}}\right|^2
\,.
\ee
Since the resulting density matrix is diagonal, we can measure in the Fourier basis without loss of information, and all that remains is to infer $r$ from samples of $\Pr(k|r)$.

To understand this distribution, we must understand the Fourier transform of a sphere, which is given by \cite{MMST96}
\be
  \sum_{x \in \sphere{r}} \omega_p^{\tr{k \cdot x}}
  = \frac{G_1^d}{q} \kloos{\quadchar^d}{r}{\dist{k}/4}
\label{eq:spherefourier}
\ee
(assuming $k \ne 0$), where $G_1 = -(-1)^m \sqrt{q}$ when $p=1 \bmod 4$, and $G_1 = -(-i)^m \sqrt{q}$ when $p=3 \bmod 4$, and where we define the \emph{$\eta$-twisted Kloosterman sum}
\be
  K_\eta(a,b)
  := \sum_{c \in \F{q}} \eta(c) \, \omega_p^{\tr(ac+bc^{-1})}
\ee
for $a,b \in \F{q}$, and for any multiplicative character $\eta$ of $\F{q}^\times$.
(This exponential sum can be viewed as the discrete analog of a Bessel function.)

If the dimension is odd, then we are interested in a $\quadchar$-twisted Kloosterman sum, also known as a Sali\'e sum.  This has the explicit form \cites{Sal32,Car53}
\be
  K_\quadchar(a,b) =
  \begin{cases}
    G_1 & ab=0, \text{~$a \ne 0$ or $b \ne 0$} \\
    2 \quadchar(b) G_1 \cos\frac{4 \pi \tr \sqrt{ab}}{p} & \quadchar(ab)=1 \\
    0 & \quadchar(ab)=-1 \text{~or~} a=b=0 \,.
  \end{cases}
\ee
In particular, we see that $\Pr(\dist{k}=0)$ is exponentially small, and for $\dist{k} \ne 0$, $\quadchar(\dist{k})$ determines $\quadchar(r)$ as follows. For $r \ne 0$, if $\quadchar(r\dist{k})=-1$, $\Pr(k|r)=0$.
On the other hand, if $r=0$, $\Pr(\quadchar(\dist{k})=+1)=\Pr(\quadchar(\dist{k})=-1)=1/2 - o(1)$.
This gives a simple quantum algorithm to determine $\quadchar(r)$.
\begin{theorem}
For $d$ odd, there is an efficient bounded-error quantum algorithm to determine $\quadchar(r)$.
\end{theorem}
\begin{proof}
The algorithm repeats the following process a constant number of times: Prepare $\rho_r$, perform the Fourier transform,  measure a value of $k$, compute $\dist{k}$, and discard the result if $\dist{k}=0$.  If the results include points with both $\quadchar(\dist{k})=+1$ and $\quadchar(\dist{k})=-1$, output $r=0$.  Otherwise, output the common value of $\quadchar(\dist{k})$.  A straightforward calculation shows that this algorithm succeeds with constant probability.
\end{proof}

Ideally, we would like to determine not just $\quadchar(r)$, but rather $r$ itself.  While we do not know an efficient algorithm, we can at least show that polynomially many queries suffice:
\begin{theorem}\label{thm:odd}
For $d$ odd, $\poly(\log q)$ queries to the hidden radius oracle suffice to determine $r$.
\end{theorem}
The proof is given in \app{odd}.

If the dimension is even, then the distribution $\Pr(k|r)$ depends on the (non-twisted) Kloosterman sum
\be
  K_1(a,b)
  = \sum_{c \in \F{q}} \omega_p^{\tr(ac+bc^{-1})}
  = \sum_{c \in \F{q}} \quadchar(c^2 - 4ab) \, \omega_p^c
\,.
\ee
No closed-form expression for such sums is known.  But we do know that in the limit $q \to \infty$, the distribution of values of the Kloosterman sum asymptotically approaches the Sato-Tate (semicircle) distribution \cites{Kat88,Ado89}, and indeed the convergence to this distribution is rapid \cite{Nie91}.  Since the Sato-Tate distribution is far from uniform, this shows that the states $\rho_r,\rho_{r'}$ are information-theoretically distinguishable for typical pairs $r \ne r'$.  We conjecture that in fact arbitrary pairs can be distinguished.

Not only do we not have a closed-form expression for non-twisted Kloosterman sums, but we do not even know whether they can be efficiently approximated on a quantum computer.  If we could approximate these sums, then we could efficiently distinguish distinguishable pairs of radii.  The problem of approximately computing Kloosterman sums (as well as more general exponential sums) on a quantum computer appears to be a natural open problem.  Indeed, it will also be relevant to the even-dimension case of the problem considered in the following section.

\section{Hidden flat of centers problem}
\label{sec:hfcp}

In this section, we consider a second shifted subset problem, the \emph{hidden flat of centers problem}.  In this problem, unlike the hidden radius problem, the spheres are promised to have unit radius.  Their centers lie on an unknown flat $H$, and the goal is to determine this flat.  For a general black-box formulation of shifted subset problems that applies to the hidden flat of centers problem, see \app{ssp}.  With that black box, the classical query complexity of determining $H$ is exponential in $d \log q$.  Here we give an efficient quantum algorithm for finding $H$, provided $d=O(1)$ is odd.

Using the quantum oracle for the hidden flat of centers problem, we can produce the quantum state
\be
  \rho_H := \frac{1}{|H|} \sum_{h \in H} |\sphere{1}+h\>\<\sphere{1}+h|
\,.
\label{eq:hfcstate}
\ee
Our goal is to determine $H$ by making measurements on this state.  We do this by using a quantum walk to move amplitude from $\sphere{1}+h$ to $h$.  If we can move a sufficiently large fraction of the amplitude, then we can determine the hidden flat by (classically) solving a noisy linear algebra problem.

To move amplitude from unit spheres to their centers, we will use a continuous-time quantum walk on the \emph{Winnie Li graph}.  This graph has vertex set $\F{q}^d$, and edges between points $x,x' \in \F{q}^d$ with $\dist{x-x'}=1$.  Thus its adjacency matrix is
\be
  A := \sum_{x \in \F{q}^d} \sum_{s \in \sphere{1}} |x+s\>\<x|
\,.
\label{eq:wladj}
\ee
The continuous-time quantum walk for time $t$ is simply the unitary operator $e^{-iAt}$.  This unitary operator can be efficiently implemented on a quantum computer provided we can efficiently transform into the eigenbasis of $A$, and can efficiently compute the eigenvalue corresponding to a given eigenvector.

The adjacency matrix \eq{wladj} has eigenvectors
\be
  |\tilde k\> := \frac{1}{\sqrt{q^d}} \sum_{x \in \F{q}^d}
                 \omega_p^{\tr k \cdot x}|x\>
\ee
for $k \in \F{q}^d$, as is clear from translation invariance.  Thus we can transform to the eigenbasis of $A$ simply using the Fourier transform  \eq{fourier}.  The corresponding eigenvalues are given by the Fourier transform of a unit sphere (cf.\ \sec{hrp}):
\be
  \lambda_k =
  \sum_{x \in \sphere{1}} \omega_p^{\tr k \cdot x} =
  \begin{cases}
    |\sphere{1}| & k=0 \\
    G_1^d K(1,\dist{k}/4)/q & \text{otherwise}\,.
  \end{cases}
\ee
All of these eigenvalues are $O(\sqrt{q^{d-1}})$, with the exception of $\lambda_0 = \Theta(q^{d-1})$.  It will be helpful to remove the single large eigenvalue, so we will replace $A$ by $\bar A := A - \lambda_0 |\tilde 0\>\<\tilde 0|$.  Then we have $\norm{\bar A} \le 2 \sqrt{q^{d-1}}$ \cite{Wei48}.

\begin{lemma}
\label{lem:walk}
Suppose we start with the quantum state \eq{hfcstate}, perform the quantum walk with the modified adjacency matrix $\bar A$ for time $t=1/\sqrt{q^{d-1} \log q}$, and finally measure in the computational basis.  Then each point in $H$ occurs with probability $|H|^{-1}[1/\log q + O(1/\log^{3/2} q)]$, and any point not on $H$ occurs with probability $O(q^{-d})$.
\end{lemma}

\begin{proof}
Consider the evolution of a single sphere $|\sphere{1}+x\>$.
Taylor expanding the action of the walk, the amplitude at the center $x$ is
\ba
  \<x|e^{-i \bar A t}|\sphere{1}+x\>
  &= -it \<x| \bar A |\sphere{1}+x\> + O(\norm{\bar A}^2 t^2) \\
  &= -it \sqrt{|\sphere{1}|} [1-O(q^{-1})] + O(\norm{\bar A}^2 t^2)
\,,
\ea
so
\be
  |\<x|e^{-i \bar A t}|\sphere{1}+x\>|^2 =
  \frac{1}{\log q} + O(\log^{-3/2} q)
\,.
\ee
Averaging over $x \in H$ gives the first part of the lemma.

It remains to show that the background is nearly uniform.  To see this, note that $e^{-i \bar A t}$ leaves invariant the subspace $\text{span}\{|x\>,|\sphere{0}+x\>,|\sphere{1}+x\>,\ldots,|\sphere{q-1}+x\>\}$ (which contains the state $|\tilde 0\>$), since $A |x\> = \sqrt{|\sphere{1}|} |\sphere{1}+x\>$ and
\be
  A |\sphere{r}+x\> 
  = \frac{1}{\sqrt{|\sphere{r}|}}
    \sum_{s \in \sphere{r}} \sum_{s' \in \sphere{1}}
    |s+s'+x\>
  = \frac{1}{\sqrt{|\sphere{r}|}} \sum_{y \in \F{q}^d}
    |\sphere{1} \cap \sphere{r} + y - x| \, |y\>
\,.
\ee
Here the coefficient of $y$ depends only on $\dist{y-x}$ (with $y=x$ a special case distinct from $\dist{y-x}=0$) by the fact that the orthogonal group over $\F{q}^d$ acts transitively on nonzero points of fixed norm (as a consequence of Witt's Lemma \cite{Asc00}).  Thus the evolved state $e^{-i \bar A t} |\sphere{1}+x\>$ is spherically symmetric about $x$.

Now consider making a measurement on $|\sphere{r}+x\>$ for some $r \in \F{q}$: each point on $\sphere{r}+x$ occurs with probability $1/|\sphere{r}|$.  Averaging over $x \in H$, we see that the point $y \in \F{q}^d$ occurs with probability $|\{x \in H: y \in \sphere{r}+x\}|/(|\sphere{r}| \cdot |H|)$.  Since $|\sphere{r}| = \Theta(q^{d-1})$, $|H|=q^{\dim H}$, and the numerator is $|H \cap (\sphere{1}+y)| = O(q^{\dim H-1})$, the probability of seeing any $y \in \F{q}^d$ is $O(q^{-d})$.  Thus the piece of $e^{-i \bar A t}|\sphere{1}+x\>$ orthogonal to $|x\>$, when averaged over $x \in H$, contributes probability $O(q^{-d})$ to every point $y \in \F{q}^d$.
\end{proof}

Now we show how to reconstruct the flat $H$ using samples from this
distribution.  A priori, $\dim H$ is unknown, so we iteratively try
increasing values of $\dim H$ until the following procedure identifies
$H$.

Let $d'=\dim H + 1$, so that $d'$ points in affine general position are sufficient to determine $H$.
Suppose we sample $k = \poly(\log q)$ points, so that with high probability the number of points in $H$ is at least $4 d'$. The following lemma shows that with high probability the $k$-sample does not intersect any flat $H'$ other than $H$, of the same dimension as $H$, in more than $4d'$ points.  Thus the flat $H$ can be computed by exhaustively trying all $\binom{k}{4d'} = \poly(\log q)$ subsets of the sample points.

\begin{lemma}
Suppose we sample $k$ points independently and identically with the following distribution: the point is uniformly random in $H$ with probability at least $1/\poly(\log q)$, and any point not in $H$ has probability at most $c/q^d$ for some constant $c$.
Then $\Pr[\exists H' \neq H, \dim H' = \dim H, \text{with $\geq 4 d'$ points
from the $k$-sample}] \leq O(\binom{k}{d'}^2) (c/q)^{d'}$.
\end{lemma}
\begin{proof}[Proof sketch]
For this event to occur, either $2d'$ points must fall in $H' \cap H$
or $2d'$ points must fall in $H'-H$. We bound the probabilities of each of
these events by similar arguments.

Consider the first of these events. Let $s_1, \ldots, s_{2d'}$ be the first $2d'$ points of the $k$-sample that fall in $H' \cap H$. 
Since
$
  \Pr[\dim \Span \{ s_1, \ldots, s_{2d'}\} \leq d'-2]
  \leq (1+O(1/q)) \Pr[\dim \Span \{ s_1, \ldots, s_{2d'}\}
  \leq d'-2
  \text{~and}
  \dim \Span \{ s_1, \ldots, s_{d'}\} = \dim \Span \{ s_1, \ldots, s_{2d'}\}]
$
(where $\Span$ denotes the affine span of a set of points),
it is sufficient to bound the probability of the latter
event. For each of the $\binom{k}{d'}$ subsets 
$s_1, \ldots, s_{d'}$ within the $k$-sample, the probability of this
event is bounded by the number of ways of choosing the remaining $d'$
points out of $k$, times the probability that all
remaining $d'$ points fall in $\Span \{ s_1, \ldots, s_{d'}\}$. 
Overall the probability of this event is bounded by $\binom{k}{d'}^2
(1/q)^{d'}$. 

The case of $H'-H$ is similar; the only change is that because we have
less control over the probabilities with which points are selected,
the final bound is $\binom{k}{d'}^2 (c/q)^{d'}$. 
\end{proof}

Overall, we have shown
\begin{theorem}
Suppose $d=O(1)$ is odd.  Then there is a quantum algorithm to determine the hidden flat of centers in time $\poly(\log q)$.
\end{theorem}

Note that we could use the same algorithm for $d$ even, provided we could efficiently approximate the eigenvalues of $A$ by approximately calculating (non-twisted) Kloosterman sums.

\section{Hidden polynomial problems}
\label{sec:hpp}

In this section, we prove some general results on the distinguishability of black box functions, and then use these results to show that the quantum query complexity of the hidden polynomial problem (defined in \sec{hppspecial}) is typically polynomial.

\subsection{Distinguishability of states with given intersection properties}
\label{sec:hppgeneral}

Consider a black box function $f:X \to Y$ where $X,Y$ are finite sets. Any such function can be encoded in a quantum state using an approach analogous to the so-called standard method for the hidden subgroup problem. In this approach, we begin with a uniform superposition over the input space, compute the black box function in an auxiliary register, and then discard that register, giving (up to normalization)
\ba
\sum_{x \in X} |x\>
&\mapsto \sum_{x \in X} \sum_{x \in X} |x,f(x)\> \\
&\mapsto |\level{f}{y}\>
\text{~where $y \in Y$ occurs with probability~} |\level{f}{y}|/|X|
\,.
\ea
(Recall that $\level{f}{y} := f^{-1}(y) = \{x \in X: f(x)=y\}$ denotes the level set of $f(x)$ with value $y$.)
In other words, this procedure uses one query of the black box to produce the mixed quantum state
\be
\rho_f := \sum_{y \in Y} \frac{|\level{f}{y}|}{|X|}
|\level{f}{y}\>\<\level{f}{y}|
\,.
\label{eq:hdstate} 
\ee

Suppose that $f$ is chosen from a set $\mathcal{F}$ of possible black box functions (where $\log |\mathcal{F}|=\poly(\log |X|)$), and we would like to determine which one we have. Then we can create $t=\poly(\log|X|)$ copies of the state \eq{hdstate}, $\rho_f^{\otimes t}$, and perform a quantum measurement to attempt to determine $f$. If some such measurement succeeds with high probability, then the query complexity of the problem is polynomial. For some measurement to succeed, it suffices to show that the single-copy states are pairwise distinguishable, as measured by the quantum fidelity
\be
F(\rho,\rho') := \tr |\sqrt\rho \sqrt{\rho'}|
\,.
\ee
This follows from a result of Barnum and Knill \cite{BK02}:
\begin{theorem}\label{thm:bk}
Suppose $\rho$ is drawn from an ensemble $\{\rho_1,\ldots,\rho_N\}$, where each $\rho_k$ occurs with some fixed prior probability. Then there exists a quantum measurement that returns the outcome $k$ with probability at least $1-N\sqrt{\max_{i \ne j}F(\rho_i,\rho_j)}$.
\end{theorem}
\noindent
(In fact, by the minimax theorem, this result holds even without assuming a prior distribution for the ensemble \cite{HW06}.)
In particular, since $F(\rho^{\otimes \ell},\rho'^{\otimes \ell}) = F(\rho,\rho')^\ell$, arbitrarily small error probability $\epsilon>0$ can be achieved using $\ell \ge \lceil 2(\log N - \log \epsilon)/\log(1/\max_{i \ne j}F(\rho_i,\rho_j))\rceil$, so $\ell=\poly(\log N)$ copies suffice provided the maximum fidelity is bounded away from $1$ by at least $1/\poly(\log N)$.
Such an argument has been used to show that the query complexity of the hidden subgroup problem is polynomial \cite{EHK99}; here we give analogous results for the hidden polynomial problem.

We begin by giving two bounds on the pairwise fidelity in terms of the intersection properties of the functions.

\begin{lemma}\label{lem:comb1}
Suppose $\Pr_{y,y' \in Y}[|\level{f}{y} \cap \level{f'}{y'}| \ge \alpha] \le \beta$,
and $|\level{f}{y}| \le \delta$ for all $y \in Y$.
Then $F(\rho_f,\rho_{f'})^2 \le (\alpha^2 + \beta\delta^2)|Y|^3/|X|^2$.
\end{lemma}

\begin{proof}
By the Cauchy-Schwartz inequality applied to the singular values of $\sqrt{\rho_f}\sqrt{\rho_{f'}}$ (whose rank is clearly at most $|Y|$),
\ba
F(\rho_f,\rho_{f'})^2
&\le |Y| \tr \rho_f \rho_{f'} \\
&= \frac{|Y|}{|X|^2} \sum_{y,y' \in Y}
|\level{f}{y} \cap \level{f'}{y'}|^2 \,,
\ea
and the claim follows from the assumptions.
\end{proof}

\begin{lemma}\label{lem:comb2}
Suppose $\Pr_{y \in Y}[|\level{f}{y} \cap \level{f'}{y'}| \ge \alpha] \le \beta$ for all $y' \in Y$,
and $\gamma \le |\level{f}{y}| \le \delta$ for all $y$. Then $F(\rho_f,\rho_{f'})^2 \le \alpha |Y|^2/\gamma|X| + \beta \delta |Y|/|X|$.
\end{lemma}

\begin{proof}
Let $\Pi_\rho$ denote the projector onto the support of $\rho$. By considering the POVM with elements $\Pi_\rho,1-\Pi_\rho$ and noting that the classical fidelity of the resulting distribution is an upper bound on the quantum fidelity, we have $F(\rho,\rho')\le \sqrt{\tr{\Pi_\rho \rho'}}$. Thus
\ba
F(\rho_f,\rho_{f'})^2
&\le \frac{1}{|X|}
\sum_{y,y' \in Y} \frac{|\level{f}{y} \cap \level{f'}{y'}|^2}
{|\level{f}{y}|} \\
&\le \frac{\alpha |Y|^2}{\gamma |X|}
+ \frac{1}{|X|} \sum_{y \in Y_{\text{bad}}}
\frac{\left(\sum_{y' \in Y} |\level{f}{y} \cap \level{f'}{y'}|\right)^2}
{|\level{f}{y}|} \\
&\le \frac{\alpha |Y|^2}{\gamma |X|}
+ \frac{1}{|X|} \sum_{y \in Y_{\text{bad}}} |\level{f}{y}|
\ea
where $Y_{\text{bad}}:= \{y \in Y: |\level{f}{y} \cap \level{f'}{y'}| \ge \alpha \text{~for some~} y' \in Y\}$.
Then the claim follows from the assumptions.
\end{proof}

\subsection{Distinguishability of hidden polynomial states}
\label{sec:hppspecial}

Now we specialize to the hidden polynomial problem.
Let $h(x) \in \F{q}[x_1,\ldots,x_d]$ be a polynomial in $d$ variables over $\F{q}$ of total degree $\deg h = O(1)$.
This polynomial is hidden by a function $f:X \to Y$ where $X=\F{q}^d$ and $|Y| \ge q$, which is simply $h$ composed with an arbitrary injective function from $\F{q}$ to $Y$. In particular, the level sets of $f$ are isomorphic to the level sets of $h$.
It is important that the black box hiding function $f$ is not simply
the hidden polynomial $h$, so that the problem of reconstructing $h$
from queries to $f$ will be hard for a classical computer. However,
$\rho_f = \rho_h$ by the isomorphism of the level sets, so it is
sufficient to calculate the fidelity between the states as if the
hiding functions were in fact the polynomials.

We begin by specializing Lemmas~\ref{lem:comb1} and \ref{lem:comb2} to the case of hidden polynomial states.
Here and in what follows, the number of variables and the degrees of polynomials
are considered bounded; the notation $o(1)$ is with respect to the limit $q \to
\infty$.

\begin{corollary}\label{cor:comb1}
Let $d \ge 2$, and suppose $\Pr_{y,y' \in \F{q}}[h(x)-y \text{~and~} h'(x)-y' \text{~have a common factor}] = o(q^{-1})$. Then $F(\rho_h,\rho_{h'}) = o(1)$.
\end{corollary}

\begin{proof}
By Lemma 4.3.3 of \cite{Sch04}, provided $h$ and $h'$ do not share a common factor,
$
|\level{h}{0} \cap \level{h'}{0}|
\le q^{d-2} \deg h \deg h' \min\{\deg h,\deg h'\}
$;
thus we can take $\alpha=O(q^{d-2})$ with $\beta=o(q^{-1})$.
By the Schwartz-Zippel Lemma (Lemma 3.3.1 in \cite{Sch04}),
$
|\level{h}{0}| \le q^{d-1} \deg h
$;
thus we can take $\delta=O(q^{d-1})$.
Then the result follows from \lem{comb1}.
\end{proof}

\begin{corollary}\label{cor:comb2}
Let $d \ge 2$, and suppose $\Pr_{y \in \F{q}}[h(x)-y \text{~not absolutely irreducible}] = o(1)$.
Then for all $h'$ with $\deg h' \le \deg h$ (other than multiples of $h$), 
$F(\rho_h,\rho_{h'}) = o(1)$.
\end{corollary}

\begin{proof}
Since $h$ is irreducible, it cannot share a common factor with $h'$, so we can take $\alpha=O(q^{d-2})$ with $\beta=o(1)$.
By Lemma 5.5.1 of \cite{Sch04}, provided $h$ is absolutely irreducible, $|\level{h}{0}| = q^{d-1}[1 + O(q^{-1/2})]$, so we can take $\gamma=\Omega(q^{d-1})$ and $\delta=O(q^{d-1})$.
Then the result follows from \lem{comb2}.
\end{proof}

Finally, we show that almost all polynomials satisfy the conditions of \cor{comb2}, which implies that the query complexity of typical hidden polynomial problems is $\poly(\log q)$.

\begin{theorem}\label{thm:typical}
Fix $d \geq 2$ and $t \ge 1$.
Then for a fraction $1-o(1)$ of the polynomials $h$ in $\F{q}[x_1,\ldots,x_d]$ of total degree $t$, for all $h'$ with $\deg h' \leq t$ (other than multiples of $h$),  $F(\rho_h,\rho_{h'}) = o(1)$.
\end{theorem}

\begin{proof}
We show that the fraction of polynomials that are not absolutely
irreducible is $O(1/q)$. Then the theorem follows by application of
\cor{comb2} and Markov's inequality.

The main idea is to count nontrivial factorizations of $h$.
Let $\F{q}(d;t)$ denote the set of $d$-variate polynomials over $\F{q}$ of total degree $t$.
If $t=1$ then we know the states are distinguishable (since they are abelian hidden subgroup states), so we can assume $t \ge 2$.
It is convenient to discuss $\F{q}$-projectivized polynomials, i.e., equivalence classes with respect to multiplication by nonzero elements in $\F{q}$; denote these by $\PF{q}(d;t)$.

The number of $\F{q}$-degrees of freedom of $\PF{q}(d;t)$ (i.e., the number of elements of $\F{q}$ required to specify a member of $\PF{q}(d;t)$)
is $\binom{d+t}{d}-1$.
The number of $\F{q}$-degrees of freedom of $\PF{q^k}(d;t)$ (the set of $\F{q^k}$-projectivized polynomials with coefficients in $\F{q^k}$) is $k\big(\binom{d+t}{d}-1\big)$.
Now we rely on the following fact: Let $h \in \PF{q}(d;t)$. 
Then there is a (unique) factorization $h=h_1 \cdots h_\ell$ (for some $\ell \geq 1$) with each $h_i \in \PF{q}(d;t)$, and of the following special form: In the (unique) factorization $h_i = \eta_{i,1} \cdots \eta_{i, k_i}$ of $h_i$ over the algebraic closure of $\F{q}$, for every $j$, the smallest field containing the coefficients of $\eta_{i,j}$ is $\F{q^{k_i}}$, and the Frobenius automorphism $c \mapsto c^q$, acting on coefficients, cyclically permutes the set $\{ \eta_{i,1}, \ldots, \eta_{i, k_i} \}$.
(In particular, for any fixed $i$, the $\eta_{i,j}$ are all distinct.)
Most importantly, $\eta_{i,1}$ determines all the $\eta_{i,j}$, so the number of $\F{q}$-degrees of freedom of $h_i$ (of degree $t_i$) is  $k_i \big(\binom{d+t_i/k_i}{d}-1\big)$.

There are at most $t$ possible values for $\ell$; we bound the number of factorizations by treating $\ell>1$ and $\ell=1$ separately.

The number of $\F{q}$-degrees of freedom for the factorizations of $h$ with $\ell>1$ is upper bounded by $\max_{1 \leq t' < t} \big[\binom{d+t'}{d} + \binom{d+t-t'}{d} -2\big]$.
It suffices to show that this is $\leq \binom{d+t}{d}-2$, hence strictly less than $\binom{d+t}{d}-1$, the number of $\F{q}$-degrees of freedom of $\PF{q}(d;t)$. Fix an ordered set of size $d+t$.
It has $\binom{d+t}{d}$ subsets of size $d$, $\binom{d+t'}{d}$ subsets of size $d$ which avoid the last $t-t'$ elements, and $\binom{d+t-t'}{d}$ subsets of size $d$ which avoid the first $t'$ elements. The latter two collections have just one common element, so we need only note that for $t \geq 2$, there is at least one subset of size $d$ which is in neither collection.

The number of $\F{q}$-degrees of freedom for the factorizations with $\ell=1$ is, by the earlier discussion, $\max_{k>1} \big[k\big(\binom{d+t/k}{d}-1\big)\big]$. We again need to show that this is $\leq \binom{d+t}{d}-2$. Fix a set of size $d+t$, and partition it into $B_0$ of size $d$ and $B_1,\ldots,B_k$ each of size $t/k$. For any $1 \leq i \leq k$, the quantity $\binom{d+t/k}{d}-1$ counts the subsets of size $d$ which are contained in $B_0 \cup B_i$ but which are not equal to $B_0$. These are disjoint subsets. None of them includes $B_0$; and because $t \geq 2$, they also miss at least one other subset of size $d$, which intersects more than one of $B_1,\ldots,B_k$. Hence the desired inequality follows.
\end{proof}

\section*{Acknowledgments}

We thank Dan Abramovich for his insights on the typicality of absolute irreducibility, and John Watrous for suggesting the problem of finding an efficient quantum procedure for moving amplitude from spheres to their centers. 
AMC received support from NSF Grant PHY-0456720 and ARO Grant W911NF-05-1-0294.
LJS received support from NSF Grant PHY-0456720, ARO Grant
W911NF-05-1-0294, and NSF Grant CCF-0524828.
UVV received support from NSF Grant CCF-0524837 and
ARO Grant W911NF-07-1-0030.


\clearpage
\appendix
\section{General formulation of shifted subset problems}
\label{app:ssp}

In \sec{hrp}, we explained one way to formulate the hidden radius problem as a black box problem.  In this appendix, we give an alternative definition that applies to general shifted subset problems.  (It is also possible to give a general definition along the lines of \sec{hrp}, but such a definition requires certain intersection properties not required here.)

An instance of a shifted subset problem over $X:=\F{q}^d$ is specified by a subset of points $S$ and a set of shifts $T$.  The problem is to determine some property of $S$ or $T$ (or both) using a black box that hides the shifted subsets $S+t$ for $t \in T$.

To obfuscate the meanings of the shifts, we introduce a bijection $\tau:T \to T$.  Furthermore, to obfuscate the meanings of the points in the subsets, we introduce a bijection $\sigma_t:S \to S$ for each $t \in T$.  The
black-box function $\pi:S \times T \to X$ defined as
$
  \pi(s,t):=\tau(t)+\sigma_t(s)
$
turns an input $(s,t)$, representing an encryption of a point in the space associated with a particular shifted subset, into an explicit point $x \in X$.  We associate each encrypted shift $t \in T$ with a black-box function value $f(t)$, where $f:T \to Y$ is an injection into an arbitrary finite set $Y$.  Finally, to allow erasing the encrypted inputs $(s,t)$, we introduce the function $g:X \times Y \to (S \times T) \cup \{\varnothing\}$ defined as
\be
  g(x,y):=
  \begin{cases}
    (s,t) & \exists s\in S,t \in T: \pi(s,t)=x \text{~and~} f(t)=y \\
    \varnothing & \text{otherwise} \,.
  \end{cases}
\ee
The oracle allows us to compute $\pi$, $f$, or $g$ as desired.

Just as in \thm{classically_hard}, we have
\begin{theorem}
Any classical computation with access to $\pi$, $f$, and $g$ requires an expected exponential number of queries to obtain a $1/\poly(d \log q)$ bias for any single bit of information about $S$ or $T$.
\end{theorem}
\noindent
The proof proceeds along the same lines as before.

However, on a quantum computer, we can prepare quantum states that encode $S$ and $T$.  We begin with a uniform superposition over the encrypted inputs $(s,t)$, compute the point $x=\pi(s,t)$, compute $f(t)$, uncompute the original inputs, and finally discard the function value.  This procedure results in the state (up to normalization)
\ba
  \sum_{s \in S, t \in T} |s,t\>
  &\mapsto \sum_{s \in S,\, t \in T} |s,t,\pi(s,t)\> \\
  &\mapsto \sum_{s \in S,\, t \in T} |s,t,\pi(s,t),f(t)\> \\
  &\mapsto \sum_{s \in S,\, t \in T} |\pi(s,t),f(t)\> \\
  &\mapsto |S+t\> \text{~where $t$ is uniformly random in $T$}
\,.
\ea
In other words, we have prepared the shifted subset state
\be
  \rho_{S,T} := \frac{1}{|T|} \sum_{t \in T} |S+t\>\<S+t|
\label{eq:ssstate}
\,.
\ee

Note that we may allow the possible sets $S$ to have different sizes, and similarly for the possible sets $T$.  For example, we see from \eq{spheresizes} that spheres of nonzero radius have two different sizes in odd dimensions.  In such cases the black box functions can be expanded to include a symbol $\varnothing$ that is returned if the input is invalid.  The above procedure can still be used provided the probability that the measurement returns the outcome $\varnothing$ is small.

\section{Query complexity of the HRP in odd dimensions}
\label{app:odd}

\begin{proof}[Proof of \thm{odd}]
The distribution of $k$ is given by
\be
  \Pr(k|r)
  = \frac{1}{q |\sphere{r}|}
  \begin{cases}
  |\sphere{r}|^2/q^{d-1} & k=0 \\
  1 & r\dist{k}=0, r \ne 0 \text{~or~} \dist{k} \ne 0 \text{~with~} k \ne 0 \\
  4 \cos^2(2\pi \tr\sqrt{r \dist{k}}/p) & \quadchar(r \dist{k})=1 \\
  0 & \quadchar(r \dist{k})=-1 \text{~or~} r=\dist{k}=0 \text{~with~} k \ne 0\,.
  \end{cases}
\label{eq:oddprob}
\ee
Now consider a pair of distinct radii $r,r'$.
We have already described an efficient algorithm to determine $\quadchar(r)$ (and in particular, to decide whether $\quadchar(r)=0$), so we can assume $r,r' \ne 0$.
If $\quadchar(r) \ne \quadchar(r')$, then the distributions they induce have nearly disjoint support, and their total variation distance is $1-o(1)$.  Otherwise, we can rescale the spheres and the measured values of $\dist{k}$ so that we are effectively distinguishing radius $1$ from some arbitrary radius $r \ne 1$ with $\quadchar(r)=1$.  The minimum total variation distance between the resulting distributions is
\ba
&
\min_{\substack{r \in \F{q} \setminus \{1\} \\ \quadchar(r)=1}}
\frac{2}{q}
\sum_{\substack{s \in \F{q} \\ \quadchar(s)=1}}
\left|\cos^2\frac{2\pi\tr\sqrt{s}}{p}
-\cos^2\frac{2\pi\tr\sqrt{rs}}{p}\right| \nn
&\qquad=
\min_{r \in \F{q}^\times \setminus \{\pm 1\}}
\frac{2}{q}
\sum_{s \in \F{q}}
\left|\cos^2\frac{2\pi\tr s}{p}
-\cos^2\frac{2\pi\tr rs}{p}\right| \\
&\qquad=
\min_{r \in \F{q}^\times \setminus \{\pm 1\}}
\frac{2}{q}
\sum_{s \in \F{q}} \frac{1}{2}
\left|\cos\frac{4\pi\tr s}{p}
-\cos\frac{4\pi\tr rs}{p}\right| \\
&\qquad\ge
\min_{r \in \F{q}^\times \setminus \{\pm 1\}}
\frac{2}{q}
\sum_{s \in \F{q}} \frac{1}{4}
\left|\cos\frac{4\pi\tr s}{p}
-\cos\frac{4\pi\tr rs}{p}\right|^2 \\
&\qquad=
\min_{r \in \F{q}^\times \setminus \{\pm 1\}}
\frac{1}{q}
\sum_{s \in \F{q}} \left(
\cos^2 \frac{4\pi\tr s}{p}
- \cos \frac{4\pi\tr s}{p} \cos \frac{4\pi\tr rs}{p}
\right) \\
&\qquad=
\frac{1}{q}
\sum_{s \in \F{q}}
\cos^2 \frac{4\pi\tr s}{p} \\
&\qquad=
\frac{1}{2}
\,.
\ea
Since an arbitrary pair of radii are statistically distinguishable with constant total variation distance, $\poly(\log q)$ samples are information-theoretically sufficient to identify an arbitrary radius.
\end{proof}

Note that the distribution \eq{oddprob} in the case $\quadchar(r \dist{k})=1$ resembles the distribution induced by a well-known single-register measurement for the dihedral hidden subgroup problem \cite{EH00}, which has resisted attempts at efficient postprocessing.


\end{document}